\documentclass[runningheads,envcountsame]{llncs}
\usepackage{graphicx}
\usepackage{epstopdf}

\usepackage[usenames,dvipsnames]{color}
\usepackage[colorlinks,citecolor=blue,urlcolor=Violet]{hyperref}
\usepackage{url}
\usepackage{amsmath, amssymb}
\usepackage{graphicx}

\newtheorem{thm}[theorem]{Theorem}

\newtheorem{prop}[theorem]{Proposition}
\newtheorem{comment}[theorem]{Comment}
\newtheorem{conject}[theorem]{Conjecture}

\newcommand{\NN}{{\mathbb{N}}}

\newcommand{\QQ}{{\mathbb{Q}}}

\newcommand{\D}{\mathcal D}
\newcommand{\FS}{\mathcal S}
\newcommand{\LL}{\mathcal L}
\newcommand{\bi}{\begin{itemize}}
\newcommand{\ei}{\end{itemize}}
\newcommand{\bc}{\begin{center}}
\newcommand{\ec}{\end{center}}

\newcommand{\abs}[1]{\left\lvert#1\right\rvert}
\newcommand{\ket}[1]{| #1 \rangle}

\newcommand{\braket}[2]{\langle #1 | #2 \rangle}
\newcommand{\braopket}[3]{\langle #1 | #2 | #3 \rangle}

\begin{document}

\titlerunning{Spectral Representation of Some C.E.\ Sets With an Application} 
\title{Spectral Representation of Some Computably Enumerable Sets With an Application to Quantum 
Provability\thanks{Partially supported by JSPS KAKENHI Grant Number 23650001.}}

\author{Cristian S. Calude\inst{1}%
\thanks{Work done in part during a visit to
Research and Development Initiative,
Chuo University, Tokyo, Japan, January 2013; partially supported also by Marie Curie FP7-PEOPLE-2010-IRSES Grant RANPHYS.}
\and Kohtaro Tadaki\inst{2}\thanks{Corresponding author.}}

\institute{Department of Computer Science, University of Auckland, Auckland,
New Zealand 
\email{cristian@cs.auckland.ac.nz}
\and
Research and Development Initiative, Chuo University,
Tokyo, Japan\\
\email{tadaki@kc.chuo-u.ac.jp}
}

\maketitle

\vspace*{-2mm}

\begin{abstract}
We propose a new type of quantum computer which is used to prove a spectral representation for a class $\FS$ of computable sets.
When  $S\in \FS$ codes the theorems of a formal system,  the quantum computer produces through measurement all theorems and  proofs of the formal system.
We conjecture that  the spectral representation is valid  for all computably enumerable sets. The conjecture implies that the theorems of a  general formal system, like Peano Arithmetic or ZFC,  can be produced through measurement; however, it  is unlikely that the quantum computer  can produce  the proofs as well, as in the particular case of $\FS$.
The analysis suggests that showing the provability of a statement is different from writing up the proof of the statement.
\end{abstract}

\section{Introduction}

\vspace*{-1mm}

Mathematical results are accepted only if they have been proved:  {\em the proof concludes with the proven statement,  the  theorem}. The proof comes first and justifies the theorem. Classically,  there is no alternative scenario.

The genius mathematician Srinivasa Ramanujan
discovered nearly 3900 results~\cite{Berndt}, many without proofs; nearly all his claims have been proven correct. Ramanujan first {\em recognised} a true statement and only later that statement was {\em proven}, hence accepted as a
{\em theorem}. While we don't know how Ramanujan's mind was able to ``discover'' mathematical true facts, we can ask whether there is a way to understand, and possibly imitate, his approach. 

In this paper a new type of quantum computer is used to prove a spectral representation for a class $\FS$ of computable sets is proved.
For every  $S\in \FS$  we 
construct a quantum system
in such a way that the elements of $S$ are exactly the eigenvalues of the Hermitian operator representing
an observable of the quantum system, i.e.\ the spectrum of the operator.
In particular, $S$ can be represented by the energy of the associated quantum system.
The operator associated   to $S\in\FS$ has a special numerical form which guarantees that
by measurement we get both the element and the proof that the element is in $S$.
 {\em We conjecture that  the spectral representation is valid  for all computably enumerable sets.}

When  $S\in \FS$ codes the theorems of a formal system, then the associated quantum computer produces through measurement the theorems  of the formal system and their proofs.
  The conjecture implies that every theorem of a  general (recursively axiomatisable)  formal system, like Peano Arithmetic or ZFC,  can be produced through measurement. However,
we argue that  in this general case the quantum procedure produces, like Ramanujan, only the true the statement, but not its proof.
  Of course,
the proof can be algorithmically generated by a classical algorithm,
albeit in a possibly very long time (such a computation makes sense only for statements recognised as ``interesting'').
For example, if the Riemann hypothesis is produced by the quantum procedure we will know that the famous hypothesis is true.
However, to have a formal proof---whose existence is guaranteed by the correctness of the quantum procedure---we may need to run a very long classical procedure.
The proof obtained in this way could be rather unsatisfactory, as it may not convey the ``understanding'',
the reason for which the Riemann hypothesis holds true (see also \cite{CM}).
Although such a proof  may not make us ``wiser''~\cite{Manin}, it may stimulate the search for better arguments.

The paper is structured as follows. In Section~\ref{QM} we present the basic quantum mechanical facts necessary for describing our quantum systems.
In Section~\ref{CS} we describe a class of computable sets for which we can prove in Section~\ref{RT} the representability theorem and
its application to quantum provability (in Section~\ref{QP}).
In Section~\ref{Cnj}  we discuss the generalisation of the quantum procedure to all computably enumerable sets and
in Section~\ref{QPP}  its application to quantum provability for arbitrary formal systems. 

\vspace*{-2.0mm}

\section{Quantum mechanical facts}
\label{QM}

\vspace*{-1.5mm}

We start with some basic facts on quantum mechanics needed for this paper.
The quantum mechanical arguments
are presented at the level of mathematical rigour adopted in quantum mechanics textbooks written by physicists,
for example, Dirac~\cite{Dirac58} and Mahan~\cite{Mahan10}.

A state of a quantum system is represented by
a vector
in a Hilbert space $\mathcal{H}$.
The vector and the space are called \emph{state vector} and \emph{state space}, respectively.
The \emph{dynamical variables} of  a system are quantities such as
the coordinates and the components of momentum and angular momentum of particles, and the energy of the system.
They play
a crucial
role not only in classical mechanics but also in quantum mechanics.
 Dynamical variables in quantum mechanics are represented by Hermitian operators on the state space $\mathcal{H}$.
A dynamical variable of the system is called an \emph{observable} if
all eigenvectors of the Hermitian operator representing it form
a complete system for $\mathcal{H}$.
Normally we assume that a measurement of any observable can be performed upon a quantum system in any state
(if we ignore the constructive matter, which is one of the points of this paper).

The set of possible outcomes of a measurement of an observable $\mathcal{O}$ of a system is
the eigenvalue spectrum of the Hermitian operator representing $\mathcal{O}$.
Let $\{\ket{m,\lambda}\}$ be
a complete orthonormal system of eigenvectors of the Hermitian operator $A$ representing
an observable $\mathcal{O}$ such that
$A\ket{m,\lambda}=m\ket{m,\lambda}$
for all eigenvalues $m$ of $A$ and all $\lambda$,
where the parameter $\lambda$ designates the degeneracy of the eigenspace of $A$.
Suppose that a measurement of $\mathcal{O}$ is performed
upon a quantum system in the state represented by a normalized vector $\ket{\Psi}\in\mathcal{H}$.
Then the probability of getting the outcome $m$ is given by $p(m)=\sum_{\lambda}\abs{\braket{m,\lambda}{\Psi}}^2$,
where $\braket{m,\lambda}{\Psi}$ denotes the inner product of the vectors  $\ket{m,\lambda}$ and $\ket{\Psi}$.
Moreover, given that the outcome $m$ occurred, the state of the quantum system immediately after the measurement is
represented by the normalized vector
\vspace*{-1mm}
\begin{equation*}
  \frac{1}{\sqrt{p(m)}}\sum_{\lambda}\braket{m,\lambda}{\Psi}\ket{m,\lambda}.
\end{equation*}
\vspace*{-3mm}

The \emph{commutator} between two operators $A$ and $B$ is defined to be $[A,B]:=AB-BA$.
Let $\mathcal{O}_1,\dots,\mathcal{O}_k$ be observables of a quantum system and
let $A_1,\dots,A_k$ be the Hermitian operators which represent  $\mathcal{O}_1,\dots,\mathcal{O}_k$, respectively.
If the Hermitian operators commute to each other, i.e., $[A_j,A_{j'}]=0$ for all $j,j'=1,\dots,k$,
then we can perform
measurements
of all $\mathcal{O}_1,\dots,\mathcal{O}_k$ simultaneously upon the quantum system in any state.
All dynamical variables which we will consider below are assumed to be observables,
and we will identify any observable with the Hermitian operator which represents it.

In this paper we consider quantum systems consisting of vibrating particles.
The simplest one is the quantum system of \emph{one-dimensional harmonic oscillator}, which consists only of one particle vibrating in one-dimensional space.
The dynamical variables needed to describe the system are just one coordinate $x$ and its conjugate momentum $p$.
The \emph{energy} of the system is an observable, called \emph{Hamiltonian}, and is defined in terms of $x$ and $p$ by
\vspace*{-1mm}
\begin{equation*}
  H=\frac{1}{2m}(p^2+m^2\omega^2 x^2),
\end{equation*}
\vspace*{-3mm}\\
where $m$ is the mass of the oscillating particle and $\omega$ is $2\pi$ times the frequency.
The oscillation of the particle is quantized by the \emph{fundamental quantum condition}
\vspace*{-2mm}
\begin{equation}\label{quantum_condition}
  [x,p]=i\hbar,
\end{equation}
\vspace*{-3mm}\\
where $\hbar$ is  \emph{Planck's constant}.
The \emph{annihilation operator} $a$ of the system is defined by
\vspace*{-1mm}
\begin{equation*}
  a=\sqrt{\frac{m\omega}{2\hbar}}\left(x+\frac{ip}{m\omega}\right).
\end{equation*}
\vspace*{-1mm}\\
Its adjoint $a^\dag$ is called a \emph{creation operator}.
The fundamental quantum condition \eqref{quantum_condition} is then equivalently rewritten as
\vspace*{-1mm}
\begin{equation}\label{acqc}
  [a,a^\dag]=1,
\end{equation}
\vspace*{-3mm}\\
and the Hamiltonian can be represented in the form
\vspace*{-0mm}
\begin{equation}\label{Hada}
  H=\hbar\omega\left(a^\dag a+\frac{1}{2}\right)
\end{equation}
\vspace*{-2mm}\\
in terms of the creation and annihilation operators.
In order to determine the values of energy possible in the system, we must solve the eigenvalue problem of $H$.
This problem is reduced to the eigenvalue problem of the observable $N:=a^\dag a$, called a \emph{number operator}.
Using the condition \eqref{acqc}, the eigenvalue spectrum of $N$ is shown to equal the set $\NN$ of all nonnegative integers.
Each eigenspace of $N$ is not degenerate, and the normalized eigenvector $\ket{n}$ of $N$ belonging to an arbitrary eigenvalue $n\in\NN$ is given by
\vspace*{-1mm}
\begin{equation}\label{evn}
  \ket{n}=\frac{(a^\dag)^n}{\sqrt{n!}}\ket{0},
\end{equation}
\vspace*{-3mm}\\
where $\ket{0}$ is the unique normalized vector up to a phase factor such that $a\ket{0}=0$.
Since $N$ is an observable,
the eigenvectors $\{\ket{n}\}$ forms a complete orthonormal system for the state space.
It follows from \eqref{Hada} that the values of energy possible in the system are
\vspace*{-1mm}
\begin{equation*}
  E_n=\hbar\omega\left(n+\frac{1}{2}\right)\raisebox{.7mm}{,}\qquad(n=0,1,2,\dotsc)
\end{equation*}
\vspace*{-3mm}\\
where the eigenvector of $H$ belonging to an energy $E_n$ is given by \eqref{evn}.

Next we consider the quantum system of \emph{$k$-dimensional harmonic oscillators}
which consists of $k$ one-dimensional harmonic oscillators vibrating independently without no interaction.
The dynamical variables needed to describe the system are $k$ coordinates $x_1,\dots,x_k$ and their conjugate momenta $p_1,\dots,p_k$.
The Hamiltonian of the system is
\vspace*{-2mm}
\begin{equation}\label{k-H}
  H=\sum_{j=1}^k\frac{1}{2m_j}(p_j^2+m_j^2\omega_j^2 x_j^2),
\end{equation}
\vspace*{-2mm}\\
where $m_j$ is the mass of the $j$th one-dimensional harmonic oscillator and $\omega_j$ is $2\pi$ times its frequency.
The vibrations of $k$ oscillators are quantized by the fundamental quantum conditions
\begin{equation}\label{k-quantum_condition}
  [x_j,p_{j'}]=i\hbar\delta_{jj'},\qquad [x_j,x_{j'}]=[p_j,p_{j'}]=0.
\end{equation}
The annihilation operator $a_j$ of the $j$th oscillator is defined by
\vspace*{-1mm}
\begin{equation*}
  a_j=\sqrt{\frac{m_j\omega_j}{2\hbar}}\left(x_j+\frac{ip_j}{m_j\omega_j}\right).
\end{equation*}
\vspace*{-3mm}\\
The adjoint $a_j^\dag$ of $a_j$ is the creation operator of the $j$th oscillator.
The fundamental quantum condition \eqref{k-quantum_condition} is then equivalently rewritten as
\vspace*{-1mm}
\begin{align}
  &[a_j,a_{j'}^\dag]=\delta_{jj'}, \label{k-acqc1}\\
  &[a_j,a_{j'}]=[a_j^\dag,a_{j'}^\dag]=0.\label{k-acqc2}
\end{align}
\vspace*{-3mm}\\
and the Hamiltonian
can be represented in the form
\vspace*{-1mm}
\begin{equation}\label{k-Hada}
  H=\sum_{j=1}^k\hbar\omega_j\left(N_j+\frac{1}{2}\right)
\end{equation}
\vspace*{-3mm}\\
where $N_j:=a_j^\dag a_j$ is the number operator of the $j$th oscillator.
In order to determine the values of energy possible in the system,
we first solve the eigenvalue problems of the number operators $N_1,\dots,N_k$.
We can do this simultaneously for all $N_j$
since the number operators commute to each other, i.e., $[N_j,N_{j'}]=0$ for all $j,j'=1,\dots,k$, due to \eqref{k-acqc1} and \eqref{k-acqc2}.
The eigenvalue spectrum of each $N_j$ is shown to equal $\NN$ using \eqref{k-acqc1}.
We define a vector $\ket{n_1,\dots,n_k}$ as the tensor product $\ket{n_1}\otimes\dots\otimes\ket{n_k}$ of  $\ket{n_1},\dots,\ket{n_k}$,
where each $\ket{n_j}$ is
defined by \eqref{evn} using $a_j$ in place of $a$.
For each $j$, the vector $\ket{n_1,\dots,n_k}$ is a normalized eigenvector of $N_j$ belonging to an eigenvalue $n_j\in\NN$, i.e.,
\begin{equation}\label{kj-evn}
  N_j\ket{n_1,\dots,n_k}=n_j\ket{n_1,\dots,n_k}.
\end{equation}
All the vectors $\{\ket{n_1,\dots,n_k}\}$ form a complete orthonormal system for the state space.
It follows from \eqref{k-Hada} that the values of energy possible in the system are
\begin{equation*}
  E_{n_1,\dots,n_k}=\hbar\sum_{j=1}^k\omega_j\left(n_j+\frac{1}{2}\right)\raisebox{.7mm}{,}\qquad(n_1,\dots,n_k=0,1,2,\dotsc)
\end{equation*}
The vector $\ket{n_1,\dots,n_k}$ is an eigenvector of $H$ belonging to an energy $E_{n_1,\dots,n_k}$.

The Hamiltonian \eqref{k-H} describes the quantum system of $k$-dimensional harmonic oscillators
where each oscillator does not interact with any others and moves independently.
In a general quantum system consisting of $k$-dimensional harmonic oscillators, each oscillator strongly 
interacts with all others.
Its Hamiltonian has the general form
\vspace*{-0mm}
\begin{equation}\label{gH}
  P(a_1,\dots,a_{k},a_1^\dag,\dots,a_{k}^\dag),
\end{equation}
\vspace*{-4mm}\\
where $a_1,\dots,a_k$ are creation operators satisfying the quantum conditions \eqref{k-acqc1} and \eqref{k-acqc2}, and
$P$ is a polynomial in $2k$ variables with coefficients of complex numbers such that \eqref{gH} is Hermitian.%
\footnote{In the monomials appearing in $P$, the order of the variables $x_1,\dots,x_{2k}$ does not matter.
However, since $a_j$ and $a_j^\dag$ do not commute,
in substituting $a_1,\dots,a_{k},a_1^\dag,\dots,a_{k}^\dag$ into the variables of $P$ the order of these operators makes a difference.
Thus, the operator \eqref{gH} makes sense only by specifying this order.}
For example, we can consider the quantum system of $k$-dimensional harmonic oscillators whose Hamiltonian is
\vspace*{-0mm}
\begin{equation*}
  H=\sum_{j}\hbar\omega_j\left(a_j^\dag a_j+\frac{1}{2}\right)+\sum_{j\neq j'}g_{jj'}a_j^\dag a_{j'}.
\end{equation*}
\vspace*{-3mm}\\
Here the \emph{interaction terms}
$g_{jj'}a_j^\dag a_{j'}$
between the $j$th oscillator and the $j'$th oscillator with a real constant $g_{jj'}$ are
added to the Hamiltonian \eqref{k-Hada}.
Note, however, that solving  exactly the eigenvalue problem of an observable in the general form of \eqref{gH} is not an easy task.

\vspace*{-1mm}

\section{A class of unary languages}
\label{CS}

\vspace*{-1mm}

In this section we introduce a class of unary languages for which the representability theorem
proven in the next section holds true.

Let $\NN^{*}$ be the set of all finite sequences $(x_{1}, \dots ,x_{m})$ with elements in
$\NN$ ($m \in \NN$; for $m=0$ we get the empty sequence $\varepsilon$). Let
\vspace*{-1mm}
\begin{equation}
\label{S}
L((x_{1}\dots x_{m}),a)=\left(\prod_{i=1}^{m} \{1^{x_{i}}\}^{*}\right) \{1^{a}\},
\end{equation}
\vspace*{-3mm}\\
for all $(x_{1}, \dots ,x_{m}) \in \NN^{*}, a \in \NN.$

\begin{thm}
Let $ \LL_{0}$ be the minimal class of
languages $\LL$ over $\{1\}$ containing the languages $\{ 1^n \}$ for every $n\in \NN$,
and which is closed under concatenation and the Kleene star operation.
\if01
satisfying the following three conditions:\\
\noindent (1) the language $\{ 1^n \} \in \LL$, for every $n\in \NN$,\\
(2) the class $\LL$ is closed under concatenation,\\
(3) the class $\LL$ is closed under the Kleene star operation.\\
\fi
Then, $\LL_{0} = \{ L((x_{1},\dots ,x_{m}),a)\mid (x_{1},\dots ,x_{m}) \in \NN^{*}, a \in \NN\}$. 
\end{thm}

\begin{proof} The class  $ \LL_{0}$ has the required properties because
$L(\varepsilon, a) = \{1^{a}\}$, the concatenation of $L((x_{1}, \dots ,x_{m}, a)$ and
$L((y_{1}, \dots ,y_{l}), b)$ is $L((x_{1}, \dots ,x_{m}), a) L((y_{1}, \dots ,y_{l}), b) =L((x_{1}, \dots ,x_{m},y_{1}, \dots ,y_{l}), a+b)$ and
the Kleene star of $L((x_{1}, \dots ,x_{m}), a)$ is $L((x_{1}, \dots ,x_{m}), a)^{*} =L((x_{1},\dots ,x_{m}, a), 0)$.
In view of (\ref{S}), $ \LL_{0}$ is included in every class $\LL$  satisfying the properties in the statement of the theorem.
\qed
\end{proof}

\begin{corollary}
The class $\LL_{0}$ coincides with  the minimal class of
languages $\LL$ over $\{1\}$ which contains the languages $\{ 1^n \}$ and
 $\{ 1^n \}^*$, for every $n\in \NN$ and which is closed under concatenation.
\end{corollary}

\begin{comment}
i) If $L$ is a finite unary language with more than one element, then $L\not\in \LL_{0}$.\\
ii) The family $\LL_{0}$ is a proper subset of the class of regular (equivalently, context-free) languages.\\
iii) The language $\{1^{p}\mid p \mbox{  is prime}\}$ is not in  $\LL_{0}$.
\end{comment}

Consider  the minimal class $\D_{0}$ of subsets of $ \mathbb{N}$ 
containing  the sets $\{b\}$, for every $b\in \mathbb{N}$, and which is closed under the sum and the Kleene star operation.
Here the sum of the sets
 $S, T$ is the set $S+T=\{a+b \mid a\in S, b\in T\}$; the Kleene star of the set $S$ is the set
 $S^{*}=\{a_{1}+ a_{2} + \dots + a_{k}\mid k \ge 0,  a_{i}\in S, 1\le  i \le k\}$.

\begin{thm}The following equality holds true:
\label{DLL}
$\LL_{0}= \{\{1^{a}\mid a \in S\}\mid S\in \D_{0}\}.$
\end{thm}

Based on the above theorem, we identify $\LL_{0}$ with $\D_{0}$ in what follows.

\vspace*{-1mm}

\section{The representation theorem}
\label{RT}

\vspace*{-1mm}

Can a set $S\in \D_{0}$ be represented as the outcomes of a quantum measurement?  We answer this question in the affirmative. First we show that the sets in $\D_{0}$ can be generated by polynomials with nonnegative integer coefficients.

\begin{prop} \label{rep} For every set $S\in \D_{0}$ there exists a polynomial with nonnegative integer coefficients $F_{S}$ in variables $x_1,\dots,x_k$
such that $S$ can be represented as:
\vspace*{-1mm}
\begin{equation}\label{range}
 S=\{F_{S}(n_1,\dots,n_k)\mid n_1,\dots,n_k\in\NN\}.
\end{equation}
\vspace*{-5mm}
\end{prop}

\begin{proof}
Suppose that $S\in \D_{0}$.
It follows from Theorem~\ref{DLL} and \eqref{S} that there exist $a_1,\dots,a_k,a\in\NN$ such that
$S=\{a_1 n_1 + \dots + a_k n_k+a\mid n_1,\dots,n_k\in\NN\}$.
Thus,
\eqref{range} holds for the polynomial $F_S(x_1,\dots,x_k)=a_1 x_1 + \dots + a_k x_k+a$.
\qed
\end{proof}

\begin{comment} There exist infinitely many sets not in $\D_{0}$ which are representable  in the form (\ref{range}).
\end{comment}

Motivated by Proposition~\ref{rep},  we  show that every set
\vspace*{-1mm}
\begin{equation}\label{rangepol}
  S=\{F(n_1,\dots,n_k)\mid n_1,\dots,n_k\in\NN\},
\end{equation}
\vspace*{-4mm}\\
where $F$ is a polynomial in $k$ variables with
nonnegative integer coefficients,
can be represented by the set of outcomes of a \emph{constructive} quantum measurement.
For this purpose, we focus  on a quantum system consisting of $k$-dimensional harmonic oscillators whose Hamiltonian has the form
\vspace*{-0mm}
\begin{equation}\label{FH}
  H=F(N_1,\dots,N_k),
\end{equation}
\vspace*{-5mm}\\
where $N_1,\dots,N_k$ is the number operators defined by $N_j=a_j^\dag a_j$ with the annihilation operator $a_j$ of the $j$th oscillator.
Note that the substitution of $N_1,\dots,N_k$ into the variables of $F$ is unambiguously defined
since the number operators $N_1,\dots,N_k$ commute to each other.
This type of Hamiltonian is a special case of \eqref{gH}.

We say
an observable
of the form \eqref{gH} is \emph{constructive} if all coefficients of $P$ are in the form of $p+qi$ with $p,q\in\QQ$.
Thus, the Hamiltonian \eqref{FH} is constructive by definition.
Actually,  a measurement of the Hamiltonian \eqref{FH} can be performed \emph{constructively} in an intuitive sense.
The constructive measurement
consists of the following two steps:
First, the simultaneous measurements of the number operators $N_1,\dots,N_k$ are performed upon the quantum system
to produce the outcomes $n_1,\dots,n_k\in\NN$ for $N_1,\dots,N_k$, respectively.
This is possible since the number operators commute to each other.
Secondly, $F(n_1,\dots,n_k)$ is calculated and is regarded as the outcome of the measurement of the Hamiltonian \eqref{FH} itself.
This is constructively possible since $F$ is a polynomial with integer coefficients.
Thus, the whole measurement process is constructive in an intuitive sense too.

\begin{thm}\label{representability}
For every set $S$ of the form \eqref{rangepol}
there exists a constructive Hamiltonian $H$ such that the set of all possible outcomes of a measurement of $H$ is $S$.
\end{thm}
\begin{proof}
Consider the Hamiltonian $H$ of the form \eqref{FH}.
It is constructive, as we saw above.
We show that the eigenvalue spectrum of $H$ equals to $S$.

First, using \eqref{kj-evn} we get
\vspace*{-1mm}
\begin{equation}\label{F-evn}
  F(N_1,\dots,N_k)\ket{n_1,\dots,n_k}=F(n_1,\dots,n_k)\ket{n_1,\dots,n_k}
\end{equation}
\vspace*{-3mm}\\
for every $n_1,\dots,n_k\in \NN$.
Thus, every element of $S$ is an eigenvalue of $H$.
Conversely, suppose that $E$ is an arbitrary eigenvalue of $H$.
Then there exists a nonzero vector $\ket{\Psi}$ such that $H\ket{\Psi}=E\ket{\Psi}$.
Since all vectors $\{\ket{n_1,\dots,n_k}\}$ form a complete orthonormal system for the state space,
there exist complex numbers $\{c_{n_1,\dots,n_k}\}$ such that
$\ket{\Psi}=\sum_{n_1,\dots,n_k}c_{n_1,\dots,n_k}\ket{n_1,\dots,n_k}$.
It follows from \eqref{F-evn} that
\vspace*{-2mm}
$$\sum_{n_1,\dots,n_k}c_{n_1,\dots,n_k}F(n_1,\dots,n_k)\ket{n_1,\dots,n_k}=\sum_{n_1,\dots,n_k}c_{n_1,\dots,n_k}E\ket{n_1,\dots,n_k}.$$
\vspace*{-3mm}\\
Since the vectors $\{\ket{n_1,\dots,n_k}\}$ are independent, we have
\vspace*{-0mm}
\begin{equation}\label{cind}
  c_{n_1,\dots,n_k}(E-F(n_1,\dots,n_k))=0,
\end{equation}
\vspace*{-4mm}\\
for all $n_1,\dots,n_k\in\NN$.
Since $\ket{\Psi}$ is nonzero, $c_{\bar{n}_1,\dots,\bar{n}_k}$ is also nonzero for some $\bar{n}_1,\dots,\bar{n}_k\in\NN$.
It follows from \eqref{cind} that $E=F(\bar{n}_1,\dots,\bar{n}_k)$.
\qed
\end{proof}

\vspace*{-2.5mm}

\section{An application to quantum provability}
\label{QP}

\vspace*{-0.5mm}

Let $S$ be a set of the form \eqref{rangepol}.
In the proof of Theorem~\ref{representability},
we consider the measurement of the Hamiltonian of the form \eqref{FH}.
In the case where the state $\ket{\Psi}$
over which
the measurement of the Hamiltonian is performed is chosen randomly,
an element of $S$ is generated randomly as the measurement outcome.
In this manner, by infinitely many repeated measurements  we get exactly the set $S$.

If the set $S$ codes the ``theorems'' of a formal system $\FS$---which is possible as $S$ is computable---then
$F(n_1,\dots,n_k)\in S$ is a \emph{theorem} of $\FS$ and
the numbers $n_1,\dots,n_k$ play the role of the \emph{proof} which certifies it.

Suppose that a \emph{single} measurement of the Hamiltonian of the form \eqref{FH} was performed upon a quantum system in a state
 represented by a normalized vector $\ket{\Psi}$
to produce an outcome $m\in S$, i.e., a theorem.
Then, by the definition of theorems, there exists a proof  $n_1,\dots,n_k$ which makes $m$ a theorem, i.e., which satisfies $m=F(n_1,\dots,n_k)$.
Can we extract the proof $n_1,\dots,n_k$ after the measurement?
This can be possible in the following manner:
Immediately after the measurement, the system is in the state represented by the normalized vector $\ket{\Phi}$ given by
\vspace*{-2mm}
\begin{equation*}
  \ket{\Phi}=\frac{1}{\sqrt{C}}\sum_{m=F(n_1,\dots,n_k)}\braket{n_1,\dots,n_k}{\Psi}\ket{n_1,\dots,n_k},
\end{equation*}
\vspace*{-3mm}\\
where $C$ is
the probability of getting the outcome $m$ in the measurement given:
\vspace*{-0mm}
\begin{equation*}
  C=\sum_{m=F(n_1,\dots,n_k)}\abs{\braket{n_1,\dots,n_k}{\Psi}}^2.
\end{equation*}
\vspace*{-3mm}\\
Since the number operators $N_1,\dots,N_k$ commute to each other,
we can perform the simultaneous measurements of $N_1,\dots,N_k$ upon the system in the state $\ket{\Phi}$.
Hence, by performing the measurements of $N_1,\dots,N_k$,
we obtain any particular outcome $n_1,\dots,n_k$ with probability $\abs{\braket{n_1,\dots,n_k}{\Phi}}^2$.
Note that
\vspace*{-0mm}
\begin{equation*}
  \sum_{m=F(n_1,\dots,n_k)}\abs{\braket{n_1,\dots,n_k}{\Phi}}^2=\sum_{m=F(n_1,\dots,n_k)}\abs{\braket{n_1,\dots,n_k}{\Psi}}^2/C=1.
\end{equation*}
\vspace*{-3mm}\\
Thus, with probability one we obtain some outcome $n_1,\dots,n_k$ such that $m=F(n_1,\dots,n_k)$.
In this manner we can
immediately
extract the proof $n_1,\dots,n_k$ of the theorem $m\in S$  obtained
as a measurement outcome.

\vspace*{-2mm}

\section{A conjecture}
\label{Cnj}

\vspace*{-1mm}

In the early 1970s, Matijasevi\v{c}, Robinson, Davis, and Putnam solved negatively Hilbert's tenth problem by proving the MRDP theorem (see Matijasevi\v{c} \cite{Matijasevic93} for details)
which states that every computably enumerable subset of $\NN$ is Diophantine.
A subset $S$ of $\NN$ is called \textit{computably enumerable} if there exists a (classical) Turing machine that,
when given $n\in\NN$ as an input, eventually halts if $n\in S$ and otherwise runs forever.
A subset $S$ of $\NN$ is \textit{Diophantine} if
there exists a polynomial $P(x,y_1,\dots,y_k)$ in variables $x,y_1,\dots,y_k$ with integer coefficients such that,
for every $n\in\NN$, $n\in S$ if and only if there exist $m_1,\dots,m_k\in\NN$ for which $P(n,m_1,\dots,m_k)=0$.\\[-1ex]

Inspired by the MRDP theorem, we conjecture the following:

\vspace*{-0mm}

\begin{conject}\label{T}
For every computably enumerable subset $S$ of $\NN$,
there exists
a constructive observable $A$ of the form of \eqref{gH} whose eigenvalue spectrum equals $S$.
\end{conject}

\vspace*{-1mm}

Conjecture~\ref{T} implies that when we perform a measurement of the observable $A$,
 a member of the computably enumerable $S$ is stochastically obtained as a measurement outcome.
As we indefinitely repeat  measurements of $A$, members of $S$ are being enumerated, just like a Turing machine enumerates $S$.

In this way 
a new type of quantum mechanical computer is postulated to exist. How can we construct it?
Below we discuss some properties of this hypothetical quantum computer.

As in the proof of the MRDP theorem---in which a whole computation history of a Turing machine is encoded in
(the base-two expansions of) the values of variables of a Diophantine equation---{\em a whole computation history of a Turing machine is
encoded in a single quantum state which does not make time-evolution
(in the Schr\"odinger picture).}
Namely, a whole computation history of the Turing machine $M$ which recognises $S$ is encoded in an eigenstate of the observable $A$ which
is designed appropriately using the creation and annihilation operators.
To be precise, let $\ket{\Psi}=\sum_{n_1,\dots,n_k}c_{n_1,\dots,n_k}\ket{n_1,\dots,n_k}$ be an eigenvector of $A$ belonging to an eigenvalue
$n\in S$
such that each coefficient $c_{n_1,\dots,n_k}$ is drawn from a certain finite set of complex numbers including $0$ and
the set $\{(n_1,\dots,n_k)\mid c_{n_1,\dots,n_k}\neq 0\}$  is finite.
The whole computation history of $M$ with the input
$n$
is encoded
in the coefficients $\{c_{n_1,\dots,n_k}\}$ of $\ket{\Psi}$
such that each finite subset obtained by dividing appropriately $\{c_{n_1,\dots,n_k}\}$ represents
the configuration (i.e., the triple of the state, the tape contents, and the head location)
of the Turing machine $M$
at the corresponding time step.
The observable $A$ is  constructed such that its eigenvector encodes the whole computation history of $M$,
using the properties of the creation and annihilation operators such as
\begin{equation*}
  a_j^\dag\ket{n_1,\dots,n_{j-1},n_j,n_{j+1},\dots,n_k}=\sqrt{n_j+1}\ket{n_1,\dots,n_{j-1},n_j+1,n_{j+1},\dots,n_k},
\end{equation*}
by which the different time steps are connected in the manner corresponding to the Turing machine computation of $M$.
In the case of $n\notin S$, the machine $M$ with the input $n$ does not halt.
This implies that the length of the whole computation history is infinite and
therefore the set $\{(n_1,\dots,n_k)\mid c_{n_1,\dots,n_k}\neq 0\}$ is infinite,
which results in that the norm of $\ket{\Psi}$ being indefinite and hence $\ket{\Psi}$ not being an eigenvector of $A$.
In this manner, any eigenvalue of $A$ is limited to a
member~of~$S$.

Note that there are many computation histories of a Turing machine depending on its input.
In the proposed quantum mechanical computer,
the measurement of $A$ chooses one of the computation histories stochastically and
the input corresponding to the computation history is obtained as a measurement outcome.
The above analysis shows that  Conjecture~\ref{T} is likely to be true.

The main feature of the proposed quantum mechanical computer is that {\em the evolution of computation  does not correspond to the time-evolution of the underlying quantum system}. Hence, in contrast with a conventional quantum computer,  the evolution of computation does not have to form a unitary time-evolution, so  it is not negatively influenced by \textit{decoherence}\footnote{Decoherence, which is induced by the interaction of quantum registers with the external environment, destroys the superposition of states of the quantum registers,
which plays an essential role in a conventional quantum computation.}, a serious obstacle to the physical realisation of a conventional quantum computer.

Again, in contrast with a conventional quantum computer, this proposed
quantum mechanical computer can be physically realisable even as a solid-state device at room temperature (the lattice vibration of solid crystal, i.e., \emph{phonons}),
which strongly interacts with the external environment.
A member of
$S$
is obtained as a measurement outcome in an instant
by measuring the observable $A$.
For example, in the case when the observable $A$ is the Hamiltonian of a quantum system, the
measurement outcome corresponds to the energy of
the system.
In this case, we can probabilistically decide---with sufficiently small error probability---whether a given
$n\in\NN$ is in
$S$:
the quantum system is first prepared
in a state $\ket{\Psi}$ such that
the expectation value
$\braopket{\Psi}{A}{\Psi}$ of the measurement of the energy over $\ket{\Psi}$ is approximately $n$,
and then the measurement is actually performed.
This
computation deciding the membership of $n$ to $S$ terminates in an instant
if sufficiently high amount of energy (i.e., around $n$) is
pumped.

\vspace*{-2mm}

\section{Quantum proving without giving the proof}
\label{QPP}

\vspace*{-2mm}

In Section~\ref{QP} we discussed the quantum provability for a formal system whose theorems can be coded
by a set $S$ defined as in \eqref{rangepol}. 
When an element $m$ is obtained as an outcome of the measurement,
we can extract the proof $n_1,\dots,n_k$ which certifies that $m$ is a theorem of the formal system $\FS$, i.e.,
it satisfies $m=F(n_1,\dots,n_k)$,
by performing the second measurement
over the state
immediately after the first measurement.

Actually, the proof $n_1,\dots,n_k$ may be generated slightly before the theorem $F(n_1,\dots,n_k)$ is obtained, like in the classical scenario.
As we saw in Section~\ref{RT}, the measurement of $F(N_1,\dots,N_k)$ can first  be performed by simultaneous measurements of the number operators $N_1,\dots,N_k$ to produce the outcomes $n_1,\dots,n_k\in\NN$; then, the  theorem  $m=F(n_1,\dots,n_k)$, classically calculated from $n_1,\dots,n_k$, can be regarded as the outcome of the measurement of $F(N_1,\dots,N_k)$ itself.

In general, the set of all theorems of a (recursively axiomatisable) formal system, such as Peano Arithmetic or ZFC, forms a computably enumerable set and not  a computable set of the form \eqref{rangepol}.
In what follows,
we argue the plausibility that, for general formal systems, the proof cannot be obtained immediately after the theorem was obtained
via the quantum procedure proposed
in the previous section.

Fix a formal system whose theorems form a computably enumerable set. As before we identify a formula with a natural number.
Let $M$ be a Turing machine such that, given a formula $F$ as an input,
$M$ searches all proofs one by one and halts if $M$ finds the proof of $F$.
Assume that Conjecture~\ref{T} holds.
Then there exists an observable $A$ of an infinite dimensional quantum system such that
$A$ is constructive and the eigenvalue spectrum of $A$ is exactly the set of all provable formulae.
Thus, we obtain a provable formula as a measurement outcome each time we perform a measurement of $A$;
it is stochastically determined which provable formula is obtained.
The probability of getting a specific provable formula $F$ as a measurement outcome depends on the choice of the state $\ket{\Psi}$ on which we perform the measurement of $A$. In some cases the probability can be
 very low,
and therefore we may be able to get the provable formula $F$ as a measurement outcome only once,
even if we repeat the measurement of $A$ on $\ket{\Psi}$ many times.
 
Suppose that, in this manner, we have performed the measurement of $A$ once and then we have obtained a specific provable formula $F$ as a measurement outcome.
Then, where is the proof of $F$?
In the quantum mechanical computer discussed in Section~\ref{Cnj}, the computation history of the Turing machine $M$ is encoded in an eigenstate of the observable $A$, hence the proof of $F$ is encoded in the eigenstate of $A$,
which is the state of the underlying quantum system immediately after the measurement.

Is it possible to extract the proof of $F$ from this eigenstate?
In order to extract the proof of $F$ from this eigenstate, it is necessary to perform an additional measurement on this eigenstate. 
However, it is impossible to determine the eigenstate
in terms of the basis $\{\ket{n_1,\dots,n_k}\}$
completely by a \emph{single} measurement due the principle of quantum mechanics.
In other words,
there does not exist a POVM measurement which can determine
all the expansion coefficients $\{c_{n_1,\dots,n_k}\}$ of the eigenstate
with respect to
the basis $\{\ket{n_1,\dots,n_k}\}$ up to a global factor with nonzero probability.
This eigenstate is destroyed after the additional measurement and therefore we cannot perform any measurement on it any more.
We cannot copy the eigenstate prior to the additional measurement due to the no-cloning theorem (see \cite{BH});  and even if we start again from the measurement of $A$, we
may have little chance of getting the same provable formula $F$ as a measurement outcome.

The above analysis suggests that even if we get a certain provable formula $F$ as a measurement outcome through the measurement of $A$ it is  very difficult or unlikely to simultaneously obtain the proof of $F$.%
\footnote{For the formal system $\mathcal{S}$  in Section~\ref{QP},
we can obtain a theorem and its proof simultaneously via measurements
since the observable $F(N_1,\dots,N_k)$ whose measurements produce ``theorems'' is a function of
the commuting observables $N_1,\dots,N_k$ whose measurements produce ``proofs''.
However, this is
unlikely to be true   for general formal systems.}
This argument
 suggests that {\em for a general formal system proving that a formula is a theorem
 is different from writing up the proof of the formula.} Of course, since
  $F$ is provable, there is a proof of $F$, hence the Turing machine $M$ with the input $F$  will eventually produce that proof.
However, this classical computation may take a long time in contrast with the fact---via the measurement of $A$---it took only a moment to know that the formula $F$ is provable.

As mathematicians guess true facts for no apparent reason we can speculate that human intuition might work as in the above described quantum scenario.
As the proposed quantum mechanical computer can operate at room temperature
it may be even possible that a similar quantum mechanical process works in the human brain those offering an argument in favour of the quantum mind  hypothesis
\cite{RPSH}. The argument against this proposition according to which quantum systems in the brain decohere quickly and cannot control brain function (see \cite{MT}) could be less relevant as decoherence plays no role in the quantum computation discussed here.\bigskip

{\bf Acknowledgement.}
We thank Professor K.~Svozil for useful comments.

\vspace*{-2.5mm}


\end{document}